\pdfoutput=1

\documentclass{llncs}
\usepackage[utf8]{inputenc}
\usepackage{float}
\usepackage{color}
\usepackage{graphicx}
\usepackage{svg}
\usepackage{amssymb, amsmath, amsfonts}
\usepackage{threeparttable}
\usepackage{comment}
\usepackage[plain]{algorithm2e}
\usepackage{algpseudocode}
\usepackage{mathtools}
\usepackage{bigdelim}
\usepackage{booktabs}
\usepackage{todonotes}

\usepackage{colortbl}
\usepackage{xfrac}
\usepackage{mleftright}
\usepackage{nth}
\usepackage[justification=centering]{subcaption}
\usepackage{breakcites}
\usepackage{pdfpages}
\usepackage[n, advantage , operators ,
sets ,
adversary ,
landau ,
probability ,
notions ,
logic ,
ff ,
mm,
primitives ,
events ,
complexity ,
asymptotics ,
keys
]{cryptocode}

\usepackage{enumitem,booktabs}
\usepackage[referable]{threeparttablex}
\renewlist{tablenotes}{enumerate}{1}
\makeatletter
\setlist[tablenotes]{label=\tnote{\alph*},ref=\alph*,itemsep=\z@,topsep=\z@skip,partopsep=\z@skip,parsep=\z@,itemindent=\z@,labelindent=\tabcolsep,labelsep=.2em,leftmargin=*,align=left,before={\footnotesize}}
\makeatother

\usepackage{tikz}
\usepackage{pgfplots}
\usepackage{tikzscale}
\usetikzlibrary{arrows}
\usetikzlibrary{calc}
\usetikzlibrary{decorations.pathreplacing}
\usetikzlibrary{math}
\usetikzlibrary{matrix}
\usetikzlibrary{patterns}
\usetikzlibrary{shadows}
\usetikzlibrary{shapes}
\usetikzlibrary{positioning}
\usetikzlibrary{arrows.meta}
\usetikzlibrary{fit,positioning,automata,decorations.pathreplacing,patterns,decorations.pathmorphing}

\PassOptionsToPackage{hyphens}{url}
\usepackage{hyperref}
\hypersetup{colorlinks=true,
            urlcolor=blue,
            linkcolor=black,
            citecolor=blue,
            breaklinks=true}
\usepackage{cleveref}

\floatstyle{boxed}
\newfloat{scheme}{ht}{sc}
\floatname{scheme}{Scheme}
\Crefname{scheme}{Scheme}{Schemes}

\newfloat{functionality}{ht}{fc}
\floatname{functionality}{Functionality}
\Crefname{functionality}{Functionality}{Functionalities}

\newcommand{\F}{\ensuremath{\mathbb{F}}\xspace}
\newcommand{\G}{\ensuremath{\mathbb{G}}\xspace}
\newcommand{\N}{\ensuremath{\mathbb{N}}\xspace}
\newcommand{\R}{\ensuremath{\mathbb{R}}\xspace}
\newcommand{\Z}{\ensuremath{\mathbb{Z}}\xspace}

\newcommand{\he}{\ensuremath{\mathsf{HE}}}
\newcommand{\evk}{\ensuremath{\mathsf{evk}}}

\newcommand{\centercell}[1]{\multicolumn{1}{c}{#1}}
\newcommand{\centercellR}[1]{\multicolumn{1}{c|}{#1}}

\newcommand{\shake}{\textsc{SHAKE128}\xspace}

\renewcommand{\vec}[1]{\boldsymbol{#1}}

\newcommand*{\numero}{n\kern-.1em \raise.7ex\vbox{\hbox{\tiny \ensuremath{\circ}}\kern.5pt}}

\usepackage{amsmath,bm}
\usepackage{multirow,booktabs}
\usepackage{pifont}
\usepackage{subcaption}
\usepackage{xspace}
\usepackage{todonotes}
\usepackage{csquotes}
\usepackage{float}
\usepackage{numprint}
\npdecimalsign{.}
\npthousandsep{,}
\usepackage{dblfloatfix}

\usepackage{siunitx}
\usepackage[n,advantage,operators,sets,adversary,landau,probability,notions,logic,ff,mm,primitives,events,complexity,asymptotics,keys]{cryptocode}

\expandafter\def\expandafter\UrlBreaks\expandafter{\UrlBreaks
    \do\a\do\b\do\c\do\d\do\e\do\f\do\g\do\h\do\i\do\j%
    \do\k\do\l\do\m\do\n\do\o\do\p\do\q\do\r\do\s\do\t%
    \do\u\do\v\do\w\do\x\do\y\do\z\do\A\do\B\do\C\do\D%
    \do\E\do\F\do\G\do\H\do\I\do\J\do\K\do\L\do\M\do\N%
    \do\O\do\P\do\Q\do\R\do\S\do\T\do\U\do\V\do\W\do\X%
    \do\Y\do\Z}

\newcommand{\ifc}[1]{\textsf{#1}}
\newcommand{\env}{\mathcal{E}}
\newcommand{\simu}{\mathcal{S}}

\newcommand{\diag}{\texttt{diag}}
\newcommand{\rot}{\texttt{rot}}

\usepackage{tikz}
\usepackage{bbm}

\usetikzlibrary{arrows}
\usetikzlibrary{calc}
\usetikzlibrary{circuits.logic.US}
\usetikzlibrary{positioning}
\usetikzlibrary{shadows}
\tikzset{>=stealth}
\tikzset{font={\fontsize{9pt}{12}\selectfont}}

\title{\huge Privately Connecting Mobility to Infectious Diseases via Applied Cryptography}

\author{Alexandros Bampoulidis\inst{4} \and Alessandro Bruni\inst{3} \and Lukas Helminger\inst{1,2} \and Daniel Kales\inst{1} \and Christian Rechberger\inst{1} \and Roman Walch\inst{1,2}}
\institute{Graz University of Technology, Graz, Austria\\\email{firstname.lastname@iaik.tugraz.at} \and Know-Center GmbH, Graz, Austria \and Katholieke Universiteit Leuven, Leuven, Belgium\\\email{alessandro.bruni@kuleuven.be} \and Research Studio Data Science, RSA FG, Vienna, Austria\\\email{alex.bampoulidis@gmail.com}}

\begin{document}

\maketitle

\begin{abstract}
    {Recent work has shown that cell phone mobility data has the unique potential to create accurate models for human mobility and consequently the spread of infected diseases~\cite{wesolowski2012quantifying}.
        While prior studies have exclusively relied on a mobile network operator's subscribers' aggregated data in modelling disease dynamics, it may be preferable to contemplate aggregated mobility data of infected individuals only. Clearly, naively linking mobile phone data with health records would violate privacy by either allowing to track mobility patterns of infected individuals, leak information on who is infected, or both. This work aims to develop a solution that reports the aggregated mobile phone location data of infected individuals while still maintaining compliance with privacy expectations. To achieve privacy, we use homomorphic encryption, validation techniques derived from zero-knowledge proofs, and differential privacy. Our protocol's open-source implementation can process eight million subscribers in 70 minutes.}
\end{abstract}


\section{Introduction}\label{sec:intro}
Human mobility plays a crucial role in infectious disease dynamics. It leads to more contact between receptive and infected individuals and may introduce pathogens into new geographical regions. Both cases can be responsible for an increased prevalence or an outbreak of an infectious disease~\cite{wesolowski2016connecting}. In particular, human travel history has been shown to play a critical role in the propagation of infectious diseases like influenza~\cite{ferguson2005strategies} or measles~\cite{grenfell2001travelling}. Therefore, understanding the spatiotemporal dynamics of an epidemic is closely tied to understanding the movement patterns of infected individuals.

Until a few years ago, researchers had to rely on general data, such as relative distance and population distribution, to model human mobility. This model was then combined with a transmission model of a particular disease into an epidemiological model, which was then used to improve the understanding of the geographical spread of epidemics.
Mobile phones and their location data have the unique potential to improve these epidemiological models further. Indeed, recent work~\cite{wesolowski2012quantifying} has shown that substituting this inaccurate mobility data with mobile phone data leads to significantly more accurate models. Integrating such up-to-date mobility patterns allowed them to identify hotspots with a higher risk of contamination, enabling policymakers to apply focused measures.

While prior studies have exclusively relied on aggregated data of all mobile network operator's subscribers' it may be preferable to contemplate aggregated mobility data of infected individuals only. Indeed, a cholera study~\cite{finger2016mobile} observed that although their model succeeded in showing that some mass gatherings had major influences in the course of the epidemic, it performed less well when the cumulative incidence is low. They speculated that demographic stochasticity could be one reason for the bad performance of their model. In other words, the infected individuals' mobility pattern may not be precisely reflected by the population's mobility if the prevalence is low. To mitigate this problem, we propose the usage of infected individuals' mobile phone data, which should lead to an improvement in the predictive capabilities of epidemiological models, especially in highly dynamic situations.

\subsection{Privacy Goals}
Our system should report a heatmap of aggregated mobile phone location data of infected individuals without revealing an individual's location or whether an individual has been infected. To that end we combine various state-of-the-art privacy-preserving cryptographic primitives to design a two-party client-server protocol for which the epidemiological researcher or a health authority inputs patients' identifiers, and the mobile network operator (MNO) inputs its subscribers' location data.

Our solution should, thereby, be able to combine both datasets without leaking the inputs to the other party. Furthermore, no party should be able to gain any information on the other dataset by cheating during protocol evaluation, e.g., by providing malicious inputs. Even if both parties follow the protocol honestly, the resulting heatmap of aggregated location data can still leak sensitive information about individuals. Thus, our protocol must also prevent this inherent output leakage while still preserving the usefulness of the resulting heatmap.

\subsection{Roadmap}
In \Cref{sec:related_work}, we discuss the relevant related work. \Cref{sec:pre} provides the necessary preliminary definitions and notations. \Cref{sec:solution} first states the problem we want to solve in this article. It then gradually develops a solution by introducing privacy protection mechanisms step by step. In \Cref{sec:security}, we perform a dedicated security and privacy analysis of our solution. \Cref{sec:impl} elaborates on the implementation of our solution as well as demonstrating the performance. \Cref{sec:conlusion} concludes with a discussion about legal considerations for an actual roll-out and how multiple parties can be included. We defer to the appendix missing proofs of our security analysis (\Cref{appendix:security}) and additional material regarding differential privacy (\Cref{appendix:diffpriv}).

\section{Related Work}\label{sec:related_work}

Numerous research directions have previously sought to model the spread of infectious diseases. Most closely related to this paper is work connecting mobile phone data to infectious diseases.

\subsection{Mobility and Infectious Diseases}\label{sec:impact_human_mobility}

Mobile phone data provides an opportunity to model human travel patterns and thereby enhance understanding of infectious diseases' transmission~\cite{wesolowski2016connecting}. Location data derived from call detail records (CDRs) -- phone calls and text messages -- have been used to understand various infectious diseases' spatial transmission better, see \Cref{tab:id-mobile-phone-data}. Each of the studies got their CDRs from one MNO, which most of the time had the largest market share and coverage. The common understanding is that biases such as Multi-SIM activity and different mobile phone usage across different geographical and socio-economic groups have a limited effect on general estimates of human mobility~\cite{wesolowski2013impact}.

\begin{table}[ht]
    \caption{Studies connecting mobile phone data to diseases.}
    \label{tab:id-mobile-phone-data}
    \centering
    \begin{tabular}{lllrrr}
        \toprule
                                         & \centercell{Disease} & \centercell{Country} & \centercell{Year of dataset} & \centercell{Subscribers} & \centercell{Period}   \\
                                         &                      &                      &                              & \centercell{(millions)}  & \centercell{(months)} \\
        \midrule
        \cite{tatem2009use}              & Malaria              & Tanzania             & 2008                         & $0.8$                    & $3$                   \\
        \cite{wesolowski2012quantifying} & Malaria              & Kenya                & 2008-09                      & $14.8$                   & $12$                  \\
        \cite{isdory2015impact}          & HIV                  & Kenya                & 2008-09                      & $14.8$                   & $12$                  \\
        \cite{wesolowski2015quantifying} & Rubella              & Kenya                & 2008-09                      & $14.8$                   & $12$                  \\
        \cite{bengtsson2015using}        & Cholera              & Haiti                & 2010                         & $2.9$                    & $2$                   \\
        \cite{tatem2014integrating}      & Malaria              & Namibia              & 2010-11                      & $1.5$                    & $12$                  \\
        \cite{wesolowski2015impact}      & Dengue               & Pakistan             & 2013                         & $39.8$                   & $7$                   \\
        \cite{finger2016mobile}          & Cholera              & Senegal              & 2013                         & $0.1$                    & $12$                  \\
        \bottomrule
    \end{tabular}
\end{table}

The most common model was to assign an individual a daily location. More concretely, each subscriber was assigned to a study area on each day based on the cell tower with the most CDRs or the last outgoing CDR. Further, the primary study area ("home") was computed for each subscriber by taking the study area where the majority of days were spent.

All of the studies emphasized that preserving individuals' privacy is mandatory.
In all cases -- to the best of our understanding -- the involved MNO anonymized the CDRs before handing them over to the health authority. In addition, we found that the MNO aggregated the CDRs in at least two cases. However, none of the studies discussed privacy definitions or the potential risk of de-identification, which is exceptionally high for location data~\cite{DBLP:journals/puc/Krumm09}. Therefore, it is hard to assess if they achieved their goal of preserving individuals' privacy.

\subsection{Exposure Notification}

Many technological approaches were developed to help reduce the spread and impact of the Covid-19 pandemic~\cite{canetti2020anonymous,chan2020pact,DAR2020100307,GA20,PR20,trieu2020epione,troncoso2020decentralized}. Most of them focus on exposure notification, where the main challenges include privacy-friendliness, scalability, and utility.
These approaches crucially rely on sizable parts of the population using smartphones, enabling Bluetooth, and installing a new app. In contrast, our proposal does not help with contact tracing, but gives potentially useful epidemiological information to health authorities without requiring people to carry around smartphones or installing an app. Indeed, any mobile phone is sufficient.

In subsequent work~\cite{DBLP:journals/iacr/HolzJMPS20}, the authors propose to use a threshold PIR-SUM protocol to allow performing privacy-preserving epidemiological modeling on top of existing contact-tracing information.
Their PIR-SUM protocol is based on a multi-server private information retrieval protocol, which is not suitable for our use case where a single entity (e.g., the mobile network operator) holds all location data. While the threshold PIR-SUM protocol can in theory be built using a single-server PIR, these protocols are significantly more expensive than the multi-server PIR they use. Furthermore, their protocols require a mix-net \cite{DBLP:conf/ccs/AbrahamPY20} to provide unlinkability of their participants messages, which already requires multiple servers, and it is not immediatly obvious how to apply their ad-hoc MPC protocol to verify the validity of queries to a single-server PIR. For single-server PIR protocols based on homomorphic encryption, our input validation procedure from \Cref{sec:input_validation} might be an alternative.

\subsection{PSI-CA and PSI-SUM Protocols}

Several works attempt to improve contact tracing by enabling users to query a database, while learning nothing more then the number of intersections using PSI-CA (Private Set Intersection Cardinality)~\cite{trieu2020epione, DBLP:journals/iacr/DittmerILOEKSS20,DBLP:conf/asiacrypt/DuongPT20} protocols, or while learning nothing except the sum of the associated values of the items in the intersection using PSI-SUM~\cite{DBLP:conf/eurosp/IonKNPSS0SY20,DBLP:conf/crypto/MiaoP0SY20} protocols.

While a PSI-SUM protocol perfectly matches our use case in theory, an application of these PSI-SUM protocols in a straightforward fashion for our main scenario in \Cref{sec:real_world_param} -- where we want to calculate the sum of vectors of length $k = 2^{15}$ for a subset of $n=2^{23}$ identifiers -- would result in impractical communication cost (multiple TB).

In~\cite{DBLP:journals/iacr/LepointPRST20}, the authors propose a method to build PSI-SUM from their \emph{PIR-with-default} primitive. This approach allows one to greatly reduce the communication to be linear in the smaller set size (the size of the queried subset of the population in our scenario).
They present two approaches, where the first one has an expensive setup phase  (multiple GB transferred for our scenario) and then has very performant queries. However, our scenario's associated values are temporal location data and would change for new protocol executions, meaning the setup phase would have to be repeated each time. Their second approach does not rely on a setup phase and -- for a database size of $n= 2^{25}$ identifiers and $t=2^{12}$ queried elements -- requires 379 MB of data transfer. However, this again only calculates the sum of a single item. A naive $k = 2^{15}$-times repetition of the approach would again result in impractical communication cost.
One could investigate if the protocol can be further optimized for large associated data, since the PSI-part of the protocol does not need to be repeated.

An additional problem of the protocol in~\cite{DBLP:journals/iacr/LepointPRST20} is, that it cannot ensure that a query does not contain an item multiple times. Applied to our use case, this leads to problems in combination with differential privacy~\cite{DBLP:conf/icalp/Dwork06}, since a larger noise needs to be added for privacy, limiting utility. The protocol in~\cite{DBLP:journals/iacr/HolzJMPS20} solves the same issue by executing a separate MPC protocol to ensure that the queries are valid and do not contain duplicates.

\subsection{Generic Multi-Party Computation}
Generic Multi-Party Computation (MPC) protocols allow several parties to securely evaluate a function without having to disclose their respective inputs. Several protocols~\cite{DBLP:conf/focs/Yao86,DBLP:conf/stoc/GoldreichMW87, DBLP:conf/crypto/DamgardPSZ12,DBLP:conf/esorics/DamgardKLPSS13} and efficient implementations for generic MPC exists~\cite{scale-mamba,DBLP:conf/ccs/Keller20}, amongst many others.

We do not use generic MPC since all efficient MPC protocols exchange data linear in the size of the computed circuit and are therefore not well suited for the large databases considered in this work. Both, secret sharing and garbled circuit based MPC, would require the (secure) transmission of the server's database (either in secret-shared form or embedded in a circuit) to the client, requiring several GB of communication (e.g., $2^{23}\times 2^{15}$ matrix of $32\,$bit integers has a size of $1$\,TB). Furthermore, the most efficient secret sharing schemes, such as the popular SPDZ~\cite{DBLP:conf/esorics/DamgardKLPSS13,DBLP:conf/crypto/DamgardPSZ12}, require a so-called Beaver triple (for multiplying values) which has to be precomputed in an expensive offline phase and can not be reused. Computing enough triples for our protocol (i.e., one triple per database entry) would require $2^{38}$ triples. This triple generation alone would require $>700$ hours and more then $1000$\,TB of communication on our hardware using the MP-SPDZ library~\cite{DBLP:conf/ccs/Keller20}.

\section{Preliminaries}\label{sec:pre}
Here we will first introduce the notations and then describe homomorphic encryption (HE) and differential privacy (DP).

We write vectors in bold lower case letters and matrices in upper case letters. We use $x_i$ to access the $i$-th element of a vector $\vec{x}$. For $m\in\N$ and $x\in\Z$, let $\vec{x}^{m}$ be defined as the vector of powers of $x$: $\vec{x}^{m} = (1, x^1, ..., x^{m-1})$. We denote by $\vec{c}\circ \vec{d}$ the element-wise multiplication (Hadamard product) of the vectors $\vec{c}$ and $\vec{d}$. For a positive integer $p$, we identify $\Z_p = \Z \cap [-p/2,p/2)$. For a real number $r$, $\lfloor r\rceil$ denotes the nearest integer to $r$.

\subsection{Homomorphic Encryption}
Homomorphic encryption (HE)~\cite{DBLP:conf/stoc/Gentry09} allows to operate on encrypted data and, thus, has the potential to realize many 2-party protocols in a privacy-preserving manner. Compared to MPC, protocols using HE usually require less data communication and only one communication round, at the cost of more expensive computations.

Modern HE schemes~\cite{DBLP:conf/innovations/BrakerskiGV12,DBLP:conf/crypto/Brakerski12,DBLP:journals/iacr/FanV12,DBLP:conf/asiacrypt/CheonKKS17,DBLP:conf/asiacrypt/ChillottiGGI16} base their security on the learning with errors~\cite{DBLP:conf/stoc/Regev05} hardness assumption and its variant over polynomial rings~\cite{DBLP:conf/eurocrypt/LyubashevskyPR10}. They allow to perform both addition and multiplication on ciphertexts. During the encryption of a plaintext, random noise is introduced into the ciphertext. This noise grows with each homomorphic operation, negligible for additions but significantly for homomorphic multiplications. Once this noise becomes too large and exceeds a threshold, the ciphertext cannot be decrypted correctly anymore. We call such a scheme, that allows evaluating an arbitrary circuit over encrypted data up to a certain depth, a somewhat homomorphic encryption scheme (SHE). The specific depth depends on the choice of encryption parameters, and choosing parameters for larger depths comes, in general, with a considerable performance penalty. In this work, we use the BFV~\cite{DBLP:conf/crypto/Brakerski12,DBLP:journals/iacr/FanV12} SHE scheme to encrypt the inputs of our protocol.

Besides semantic security, two-party protocols based on HE additionally either require the notion of circuit privacy~\cite{DBLP:conf/stoc/Gentry09}, or function~\cite{DBLP:conf/crypto/GentryHV10} (evaluation~\cite{albrecht2019homomorphic}) privacy, to hide the function applied on the ciphertexts. Function privacy is often easier to achieve than circuit privacy in practice. It requires that the outputs of evaluating different circuits homomorphically on the same encrypted data need to be indistinguishable. In other words, a party decrypting the final result of a function private HE scheme can not learn anything about the circuit applied to the input data. Like many HE schemes, BFV does not naturally provide function privacy, however, it can be added by applying noise flooding~\cite{DBLP:conf/stoc/Gentry09}.

\subsection{Differential Privacy}\label{sec:dp}

Differential privacy (DP)~\cite{DBLP:conf/icalp/Dwork06} defines a robust, quantitative notion of privacy for individuals. The main idea is that the outcome of a computation should be as independent as possible (usually defined by a privacy parameter $\epsilon$) from the data of a single input. Applied to our use case, DP makes the final heatmap independent to the contribution of any individual, preventing it from leaking sensitive information.

DP is highly compatible with existing privacy frameworks and has successfully been applied to several real-world applications. Recent work~\cite{nissim2017bridging} showed that DP satisfies privacy requirements set forth by FERPA\footnote{Family Educational Rights and Privacy Act of 1974, U.S.}. Even before this analysis, several businesses were already using DP. For example, Apple~\cite{AppleDP20} and Google~\cite{GoogleDP20} have applied differential privacy to gather statistics about their users without intruding on individual's privacy. Recently, the U.S. Census 2020 uses differential privacy as a privacy protection system~\cite{USCensusDP20}.

The most prevalent technique to achieve DP is to add noise sampled from a zero-centered Laplace distribution to the outcome of the computation. The distribution is calibrated with a privacy parameter $\epsilon$ and the global sensitivity $\Delta q$ of the computation and has the following probability density function:
\[Lap(x|b) = \frac{1}{2b}e^{-\frac{|x|}{b}},\qquad \text{with } b=\frac{\Delta q}{\epsilon}\]

To add noise to a protocol operating on integers, we discretize the Laplace distribution by rounding the sampled value to the nearest integer. For a formal definition of DP, we refer to \Cref{appendix:diffpriv}.

\section{Privacy Preserving Heatmap Protocol}\label{sec:solution}
We first describe our goal and then introduce each privacy protection mechanism step by step.

\subsection{The Desired Functionality}
Our aim is to accumulate the location data of infected individuals to create a heatmap, assisting governments in managing an epidemic. Two parties controlling two different datasets are involved: A health authority who knows which individuals are infected; and an MNO who knows the location data of their subscribers. More specifically, the MNO knows how long each of their subscribers is connected to which cell towers (CDRs are a subset of this information). Based on this mobility data, our protocol could answer several questions. One in line with epidemiological literature is to look at the individuals' mobility data in the incubation period (e.g., 5-7 days for COVID-19). The final heatmap will show areas with a higher chance of getting infected with the disease. A natural extension would be to study mobility patterns after the incubation period but before confirmation/quarantine. So our protocol is generic regarding the time unit or the granularity of location data. When discussing privacy guarantees that depend on the actual data, we will explicitly outline the chosen setting.

\subsubsection*{Protocol Description.}
If the MNO knows which of its subscribers is infected, it can do the following to create the desired heatmap:
\begin{itemize}
    \item Initialize a vector $\vec{h}$ of $k$ elements with zeros, where $k$ is the total number of cell towers. Each element of this vector corresponds to one cell tower.
    \item For each infected individual, add the amount of time it spent at each cell tower to the corresponding element of the vector $\vec{h}$.
    \item Then the vector $\vec{h}$ contains the final heatmap, i.e., $h_j$ contains the accumulated time spent of all infected individuals at cell tower $j$.
\end{itemize}

Let us rewrite this process into a single matrix multiplication. First, we encode all $N$ subscribed individuals into a vector $\vec{x} \in \Z_2^N$, with $x_i \in \Z_2$ indicating, whether the individual $i$ is infected ($x_i = 1$) or not ($x_i=0$). Then we encode the location data in a matrix $Z = \left(\vec{z_1}, \vec{z_2}, \dots, \vec{z_k}\right) \in \Z^{N\times k}$ such that the vector $\vec{z_j}$ contains all the location data corresponding to the cell tower identified by $j$. In other words, the $i$-th element of the vector $\vec{z_j}$ contains the amount of time individual $i$ spent at cell tower $j$. Now, we can calculate the heatmap as $\vec{h} = \vec{x}^T \cdot Z$.

We depict the basic protocol, involving the health authority as a client and the MNO as a server, in \Cref{fig:basic_protocol}, assuming the health authority and the MNO already agreed on identifying all subscribed individuals by indices $i\in{1,...,N}$.

\begin{figure}[ht]
    \centering
    \resizebox{0.8\textwidth}{!}{%
    \fbox{
    \pseudocode[colsep=1em,addtolength=5em]{%
    \textbf{Client (Health Authority)} \< \< \textbf{Server (MNO)} \\
    \text{Input: } \vec{x}\in\Z_2^N\< \< \text{Input: }Z\in\Z^{N\times k}\\
    \text{Output: }\vec{h}=\vec{x}^T\cdot Z \< \< \text{Output: } \vec{h} \\
    \< \quad\qquad\qquad\textbf{Input}\\
    \< \sendmessageright*{\vec{x}}\\
    [][\hline]
    \< \quad\quad\textbf{Data Aggregation}\\
    \< \sendmessageleft*{\vec{h}} \< \vec{h}\gets\vec{x}^T\cdot Z\\
    \text{Output }\vec{h}
    }
    }
    }
    \caption{Basic protocol without privacy protection.}
    \label{fig:basic_protocol}
\end{figure}

\begin{remark}[Agreeing on database indices] \label{rem:indices}
    The protocol in \Cref{fig:basic_protocol} already assumes that the two parties agree on the indices of individuals in the database. In practice, the individuals can be identified using several methods, such as phone numbers, mail addresses or government ids. We now give two options to get a mapping from a phone number to a database index, while noting that any other identifier can be trivially used instead:
    \begin{itemize}
        \item The MNO and health authority engage in a protocol for Private Set Intersection (PSI) with associated data (e.g., \cite{DBLP:conf/ccs/ChenHLR18,DBLP:conf/ccs/CongMGDILR21}). In such a protocol, the health authority and the MNO input their list of phone numbers. The health authority gets as the protocol's output the phone numbers that are in both sets, as well as the associated data from the MNO. The associated data would be the index in the database in our case.
        \item The MNO sends a mapping of all phone numbers to their database index in plain. This approach is efficient and straightforward, but it discloses all subscribed individuals to the health authority. However, this is essentially a list of all valid phone numbers in random order and does not leak anything more than the validity of that number. Still, this may be an issue in some scenarios.
    \end{itemize}
    While the PSI-based solution has some overhead compared to the plain one, the performance evaluation in~\cite{DBLP:conf/ccs/CongMGDILR21} shows that a protocol execution with $2^{22}$ MNO items and $4096$ health authority items takes about $1.4$ seconds online (excluding a precomputable offline phase taking $467$ seconds) with a total communication of $8.3$ MB -- a minor increase when looking at the overall protocol.
    While PSI-SUM protocols~\cite{DBLP:conf/eurosp/IonKNPSS0SY20,DBLP:conf/crypto/MiaoP0SY20,DBLP:journals/iacr/LepointPRST20} could be used to calculate the final heatmap without revealing which identifiers are present in the MNO's set, their additional overhead is not worth the minor privacy gain, considering that for the type of identifier we are using (phone numbers), one can often already publicly check if a phone number is associated with a mobile network operator\footnote{as an example, using services such as \url{https://dexatel.com/carrier-lookup/}, or often also just calling the number}. We therefore relax our setting to allow revealing which identifiers are present in the MNO's set to take advantage of the reduced communication of our approach compared to the full PSI-SUM approaches mentioned above.
\end{remark}

Executing the protocol described in \Cref{fig:basic_protocol} would enable the MNO to learn about infected individuals, which is a massive privacy violation. On the other hand, the health authority could query a single individual's location data by sending a vector $\vec{x}=(1,0,\dots,0)$, violating privacy. In the following, we describe our techniques to protect against these violations.

\subsection{Adding Encryption}
To protect the vector send by the health authority, and therefore who is infected and who is not, we use a HE scheme $(\kgen,\enc,\allowbreak{}\dec,\eval)$. Before executing the protocol, the health authority runs $\kgen$ to obtain a secret key $\sk$ and an evaluation key $\evk$. We assume that the MNO knows $\evk$, which is required to perform operations on encrypted data, before running the protocol.

In the updated protocol, the health authority now uses $\sk$ to encrypt the input vector $\vec{x}$ and sends the resulting ciphertext vector $\vec{c}\gets \enc_{\sk}(\vec{x})$ to the MNO. The MNO then uses $\evk$ to perform the matrix multiplication on the encrypted input vector and sends the resulting ciphertext vector $\vec{h}^*\gets\eval_{\evk}\left(\vec{c}^T\cdot Z\right)$ back to health authority. The health authority can now use $\sk$ to decrypt the result and get the final heatmap $\vec{h} = \dec_{\sk}(\vec{h}^*)$.

Informally, if the used HE scheme is semantically secure, then the MNO cannot learn which individuals are infected by the disease and which are not.

\subsection{Input Validation}
\label{sec:input_validation}
In the simple protocol, the health authority could use a manipulated input vector $\vec{x}$ to include an individual multiple times (e.g., setting the corresponding vector entry to $100$ instead of $1$). Such an individual could most likely be filtered out in the final heatmap. Since the input vector is encrypted, the MNO cannot trivially check if the vector is malicious or not. Also, comparing encrypted elements is not trivially possible in most HE schemes. However, the required check can be encoded, such that it outputs $0$ if everything is correct, and a random value otherwise. We then can add this value to the final output as a masking value which randomizes the MNO's response if the input vector is malicious. We describe how to generate this masking below.

\subsubsection*{Masking Against Non-Binary Query Vector.}\label{sec:masking_binary}
Note that the HE schemes plaintext space usually is $\Z_p$, i.e., the integers modulo a prime $p$. Therefore, the inputs to our protocol -- the vector $\vec{x}$ and the matrix $Z$ -- consist of elements in $\Z_p$. As outlined above, it is crucial to the protocol's privacy that the input vector is binary, i.e., only contains $0$s and $1$s. If this is not the case, the health authority could arbitrarily modify a single person's contribution to the overall aggregated result. It is essential for DP considerations to bound the maximum possible contribution of a single individual (sensitivity).

Since the MNO only receives an encryption of the input vector, simply checking for binary values is not an option. However, we can use similar techniques to the ones used in Bulletproofs~\cite{DBLP:conf/sp/BunzBBPWM18} to provide assurance that the query vector $\vec{x}\in\Z_p^N$ contains only binary elements. First, we will exploit the following general observation. Let $\vec{d} = \vec{x} - \vec{1}$, then $\vec{x} \circ \vec{d}$ is the zero vector iff $\vec{x}$ is binary. Note that the MNO can compute an encryption of $\vec{d}$ from the encrypted input vector. The result of the Hadamard product $\vec{x} \circ \vec{d}$ can be aggregated into a single value by calculating the inner product $\langle \vec{x}, \vec{d}\rangle$, which will again be zero if $\vec{x}$ is binary. The MNO also multiplies $\vec{x}$ with powers of a random integers $y$ to reduce the probability of the health authority cheating by letting several entries of $\vec{x}$ cancel each other out during the inner product, which gives the mask:
\begin{equation}
    \mu_\mathtt{bin'} = \langle \vec{x}, (\vec{d} \circ \vec{y}^N)\rangle \,,\label{eq:mask_bin}
\end{equation}
where $\vec{y}^{N} = (1, y^1, ..., y^{N-1})$ is $y$'s vector of powers.

For the generic case of a vector $\vec{v}$ and a randomly chosen $y$, $\langle \vec{v}, \vec{y}^N \rangle = 0$ will hold for $\vec{v} \neq 0$ only with probability $N/p$~\cite{DBLP:conf/sp/BunzBBPWM18}. Using a $\nu\,$bit modulus $p$ ($p\approx2^{\nu}$), translates to a soundness error of $\nu-\text{log}_2(N)\,$bits. For details of this calculation see \Cref{appendix:sec-masks}. In particular, if we look at $N=2^{23}, \nu=60$, parameters sufficient for small nation-states (see \Cref{sec:real_world_param}), we get $37$-bit statistical security. Standard literature suggest a statistical security parameter of at least $40$-bit; therefore, we developed a method to enhance the statistical security without significant overhead.

\subsubsection*{Boosting Soundness.}
The high-level idea is that we lower the probability of cheating successfully by a random linear combination of separate masks. Intuitively, a malicious health authority would have to guess correctly for every single mask. Thus, the soundness converges to the underlying field size in the number of terms of the linear combination. For our purpose, two terms already suffice for an appropriate security level:
\begin{align}
    \mu_\mathtt{bin} & =\langle \vec{x}, (\vec{d} \circ \vec{y_1}^N)\rangle \cdot r_1 + \langle \vec{x}, (\vec{d} \circ \vec{y_2}^N) \rangle \cdot r_2 \notag
\end{align}
where $r_1, r_2 \stackrel{\$}{\gets}\Z_p \setminus\{0\}$ are two random values.
Therefore, the statistical security level increases to $\nu-1$\,bit ($=59$\,bits for $\nu = 60$). We refer to \Cref{lem:masking_binary} in \Cref{appendix:sec-masks} for a proof of this statement.

\subsubsection*{Applying the Mask.}
Once the $\mu_\mathtt{bin}$ is calculated, it gets added to the final output of the protocol. However, if the masking value is not zero, we have to make sure that a different random value is added to each element of the output vector to prevent leaking the mask if some output vector values are known beforehand. Therefore, the final mask $\vec{\mu}$ can be calculated using a random vector $\vec{r} \stackrel{\$}{\gets}(\Z_p\setminus \{0\})^k$ as follows:
\begin{equation}
    \vec{\mu} = \mu_\mathtt{bin}\cdot \vec{r} \label{eq:mask}
\end{equation}
The final mask $\vec{\mu}$ is now equal to $\vec{0}^k$ if $\vec{x}$ is a binary vector, random otherwise.

\begin{remark}[PSI-SUM with Indices]
    So far the protocol securely implements a functionality dubbed \textit{PSI-SUM with Indices}. For completeness, we included a description and its ideal functionality in \Cref{app:psisum}.
\end{remark}

\subsection{Adding Differential Privacy}\label{sec:addingdiffpriv}
The aggregated location data can still leak information about the location data of individuals. For example, the health authority could abuse the heatmap to track an individual by just querying him or by querying him alongside individuals from a completely different area. The location data of the targeted individual would be visible as an isolated zone in the resulting heatmap. Applying DP with suitable parameters will protect against such an attack since the overall goal of DP is to decrease the statistical dependence of the final result to a single database entry. In our use case, therefore, DP achieves that it is highly unlikely to distinguish between heatmaps, in which we include a single individual in the accumulation, and heatmaps, in which we do not.

Choosing proper parameters, however, highly depends on the underlying dataset. On the one hand, the chosen $\epsilon$ should be small enough to satisfy privacy concerns; on the other hand, it should be big enough not to overflow the result with noise, creating hotspots on its own. We discuss one method to choose suitable parameters in \Cref{sec:sub-privacy}.

\subsection{Final Protocol}\label{sec:protocol}
Finally, with all measures to protect privacy in place, we present the final protocol in \Cref{fig:final_protocol}.

\begin{figure*}[!htb]
    \centering
    \resizebox{\textwidth}{!}{%
    \fbox{
    \pseudocode[colsep=1em,addtolength=5em]{%
    \textbf{Client (Health Authority)} \< \< \textbf{Server (MNO)} \\
    \text{Input: } \vec{x}\in\Z_p^N\< \< \text{Input: }Z\in\Z_p^{N\times k}\\
    \text{Output: }\text{If }\vec{x}\in\Z_2^N: \vec{h}=\vec{x}^T\cdot Z + \left(\left\lfloor Lap\left(\frac{\Delta q}{\epsilon}\right)\right\rceil\right)^k\in\Z_p^{k}\< \<\text{Output: } \bot \\
    \quad\quad\quad\quad\!\text{Otherwise: }\vec{h}\stackrel{\$}{\gets}\Z^k_p\\[][\hline]
    \< \quad\qquad\textbf{Encryption}\\
    \vec{c}\gets \enc_{\sk}(\vec{x}) \< \sendmessageright*{\vec{c}}\\
    [][\hline]
    \< \quad\qquad\textbf{Data Aggregation}\\
    \< \< \vec{h}^*\gets\eval_{\evk}\left(\vec{c}^T\cdot Z\right)\\
    [][\hline]
    \< \quad\qquad\textbf{Compute Mask}\\
    \< \<\vec{d}\gets\eval_{\evk}\left(\vec{c}-\vec{1}^{N}\right)\\
    \< \< \vec{r}\stackrel{\$}{\gets}\Z_p^k;r_1,r_2,y_1,y_2\stackrel{\$}{\gets}\Z_p\\
    \< \< \mu_\mathtt{bin} \gets\eval_{\evk}\left(\langle \vec{c}, (\vec{d} \circ \vec{y_1}^N)\rangle r_1 + \langle \vec{c}, (\vec{d} \circ \vec{y_2}^N) \rangle r_2\right)\\
    \< \< \vec{\mu} \gets\eval_{\evk}\left(\mu_\mathtt{bin}\cdot \vec{r}\right)\\
    \< \< \vec{h}^* \gets\eval_{\evk}\left(\vec{h}^*+\vec{\mu}\right)\\
    [][\hline]
    \< \quad\qquad\textbf{Differential Privacy}\\
    \<  \< \vec{\delta}\stackrel{\$}{\gets} \left(\left\lfloor Lap\left(\frac{\Delta q}{\epsilon}\right)\right\rceil\right)^k\\
    \< \sendmessageleft*{\vec{h}^*} \< \vec{h}^*\gets\eval_{\evk}\left(\vec{h}^*+\vec{\delta}\right)\\
    [][\hline]
    \< \quad\qquad\textbf{Decryption}\\
    \quad \text{Output } \vec{h} \gets \dec_{\sk}(\vec{h}^*)
    }
    }
    }
    \caption{Privacy preserving heatmap protocol.}
    \label{fig:final_protocol}
\end{figure*}

\section{Security \& Privacy Analysis} \label{sec:security}
On the one side, the protocol provides input security against a malicious MNO, i.e., even if the MNO deviates from the protocol, it cannot determine the patient's identifiers (see \Cref{sec:sub-security}). On the other side, individuals' location data are protected even against a malicious health authority, i.e., the health authority cannot track individuals (see \Cref{sec:sub-privacy})).

\subsection{Security}\label{sec:sub-security}
Two-party protocols are usually proven secure with the real-ideal world paradigm~\cite{DBLP:conf/focs/Canetti01}. Roughly speaking, one has to prove that the protocol does not leak any additional information than when computed with the help of a trusted third party. The trusted third party is modeled as an ideal functionality presented in \Cref{fig:ideal}.

\begin{figure}[h]
    \centering
    \resizebox{0.95\columnwidth}{!}{
        \fbox{
            \parbox{\columnwidth}{
                \center{\textbf{\large{$\mathcal{F}_{Hmap}$}}}\\
                \vspace{3mm}
                \raggedright
                Parameters: $t,N\in\N,\beta\in\R_{+}$. Parties $P_1$ and $P_2$.
                \begin{enumerate}
                    \item Upon receiving an input $(\ifc{input},sid,P_1,P_2,\vec{x})$ from a party $P_1$, verify that $\vec{x}\in \Z^{N}_p$, else ignore input. Next, record $(sid,P_1,P_2,\vec{x})$. Once $\vec{x}$ is recorded, ignore any subsequent inputs of the form $(\ifc{input},sid,P_1,P_2,\cdot)$ from $P_1$.
                    \item Upon receiving an input $(\ifc{input},sid,P_1,P_2,Z)$ from party $P_2$, verify that $Z\in \Z_p^{N\times *}$, else ignore input. Proceed as follows: If there is a recorded value $(sid,P_1,P_2,\vec{x},w)$, compute $\vec{h}\gets \vec{x}^TZ + \left(\left\lfloor Lap\left(\beta\right)\right\rceil\right)^k$ provided that $\vec{x}\in\Z_2^N$, otherwise $\vec{h}\stackrel{\$}{\gets}\Z^k_p$, and send $(sid,P_1,P_2,k)$ where $k$ is the number of columns of $Z$ to the adversary. Then output $(\ifc{result},sid,P_1,P_2,\vec{h})$ to $P_1$, and ignore subsequent inputs of the form $(\ifc{input},sid,P_1,P_2,\cdot)$ from $P_2$.
                \end{enumerate}
            }
        }
    }
    \caption{Ideal functionality $\mathcal{F}_{Hmap}$ of the above solution.}
    \label{fig:ideal}
\end{figure}

\subsubsection*{Semi-Honest Security.}
Before we discuss malicious security, we will show that our protocol achieves semi-honest security.

\begin{lemma}\label{lem:semi}
    Let us assume $\he$ is an IND-CPA secure homomorphic encryption scheme that provides function privacy. Then the protocol in \Cref{fig:final_protocol} securely realizes $\mathcal{F}_{Hmap}$ against static semi-honest adversaries.
\end{lemma}
The high-level idea is two-fold. Firstly, by the definition of semantic security, the MNO can not learn anything from encrypted data, hence, we reduce our protocol's security against the MNO's corruption to the semantic security of the underlying HE scheme. Second, function privacy guarantees that the health authority learns nothing more about the MNO's matrix, than what can be derived from the input $\vec{x}$ and the output $\vec{h}$. The formal proof can be found in \Cref{appendix:security}.

\subsubsection*{Malicious Security.}
Achieving simulation-based security against a malicious MNO would be similar to verified HE. While some theoretical constructions exist~\cite{DBLP:conf/esorics/LaiDPW14}, they are not practical.

Instead, we show input security against a malicious MNO, also known as one-sided simulation security. This notion has been first considered in the context of oblivious transfer~\cite{DBLP:conf/soda/NaorP01}, was then formalized~\cite{DBLP:conf/tcc/HazayL08}, and recently used~\cite{DBLP:conf/ccs/ChenHLR18} in the realm of PSI. In our protocol, one-sided simulation guarantees that the patients' identifiers are protected even in the presence of a malicious MNO (one that deviates from the protocol). For a formal definition, see \Cref{appendix:security}.

\begin{theorem}\label{thm:one-sided}
    Let us assume $\he$ is an IND-CPA secure homomorphic encryption scheme that provides function privacy. Then the protocol in \Cref{fig:final_protocol} securely realizes $\mathcal{F}_{Hmap}$ with one-sided simulation in the presence of a maliciously controlled MNO.
\end{theorem}

\begin{proof}
    From \Cref{lem:semi}, we already know that the protocol is secure against semi-honest adversaries. The only thing left to show is input privacy of the health authority against a malicious MNO, i.e., the MNO is not able to learn any information from the health authority's input (patients' identifier). Now, due to the fact that the MNO's view only includes an encryption of the health authority's input, by the semantic security of $\he$, we have that the MNO learns nothing about the health authority's input.
\end{proof}

\subsection{Privacy}\label{sec:sub-privacy}

The protocol's output exposes aggregate information, namely the amount of time spent by individuals at a cell tower, to the health authority. In the worst case, only one individual is present in the aggregation. Even in this case the health authority should not be able to single out any individual. To mitigate this threat, we propose to use \emph{differential privacy} (DP).

\subsubsection{Privacy-Utility Tradeoff.}
It is a challenge to choose the right amount of noise to protect individuals' privacy while still preserving utility. Ultimately, this tradeoff is not only technical but also has to take into account normative considerations~\cite{nissim2017bridging}. Here, we only explore the technical tradeoff and leave the latter one to policymakers.

There has been limited research specifically addressing the technical tradeoff~\cite{lee,kohli,DBLP:conf/csfw/HsuGHKNPR14}. However, the methods of~\cite{lee,kohli} are not applicable to our protocol since they require either input from the individuals~\cite{kohli} or "\textit{knowing the queries to be computed}"~\cite{lee}. Therefore, we choose to follow~\cite{DBLP:conf/csfw/HsuGHKNPR14}'s rigorous method to find real-world parameters for DP.

\subsubsection{Choosing the Right \texorpdfstring{$\epsilon$}{epsilon}.} \label{sec:epsilon}
The model in~\cite{DBLP:conf/csfw/HsuGHKNPR14} provides a principled approach to choose $\epsilon$. It can be split into two major steps resulting in two constraints that have to be satisfied simultaneously. First, one chooses the desired utility by defining a confidence interval. The parameters of the confidence interval give the first constraint on the required minimum number $w$ of infected individuals and $\epsilon$. Note that if the health authority does not provide the number of infected individuals $w$ or lies about it, privacy is not affected. Only utility cannot be guaranteed any more.

Further, the method requires setting a bound on the expected privacy harm per individual and estimating the expected cost (baseline cost) of not being part of the outcome (e.g., database breach). This leads to the second constraint on $\epsilon$. Every pair of parameters, $\epsilon$, and the number of infected individuals $w$ that simultaneously fulfill both constraints, is a reasonable privacy-utility tradeoff choice.

To illustrate this method, we now provide a possible set of values for this example. Choosing these values requires a few assumptions. We want to highlight that our assumptions are, at best, educated guesses. The real-world values have to be adjusted to the concrete circumstances and be discussed by a group of privacy, ethical, legal, epidemiological and policy experts.

First, the time unit is days for consistency with previous epidemiological studies, see \Cref{sec:impact_human_mobility}. We decided to aim for a margin of error of $\pm 5\%$ with a probability of $0.95$ (confidence). In terms of privacy, the method requires us to estimate the expected base costs (harm) that arise for an individual by using the MNO's services (data breach at the MNO would leak the location data), i.e., without even being part of the computation. We assume that without performing our protocol, this probability is less than $0.00001$. In the case of a leakage, we set the monetary harm inflicted to an individual to an exemplary amount of $\$1000$ per day. This seems reasonable since most smartphone users divulge exact location data for far less than that amount to companies. Now, we can calculate the expected baseline cost as $0.00001\cdot \$1000=\$0.01$ per day of leakage. We think performing the protocol is justified if the cost of participating does not exceed $\$0.02$. We arrive at the following two constraints (see \Cref{appendix:dp_ecnomic} for details)

\begin{align*}
    \exp\left(-0.05\cdot\frac{w\epsilon}{2}\right) & \leq 0.05\;\;\; & \text{ (utility)} \\
    0.01\cdot(e^{\epsilon}-1)                      & \leq 0.02,      & \text{ (privacy)}
\end{align*}
which are illustrated in \Cref{fig:diff_priv_epsilon}.

\begin{figure}[!htb]
    \centering
    \includegraphics[width=0.75\columnwidth]{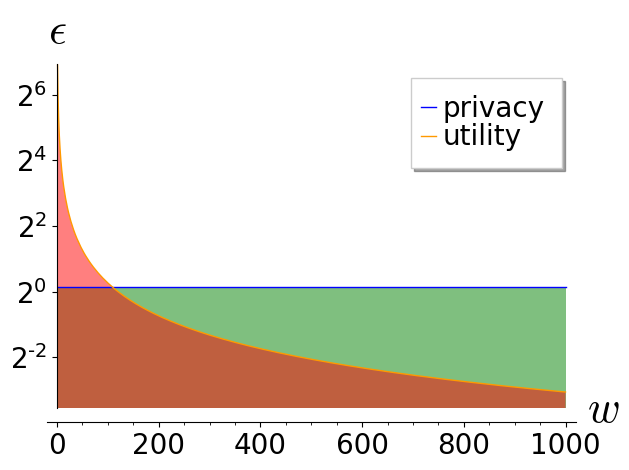}
    \caption{Privacy-utility tradeoff: The green area are possible combination for $\epsilon$ and $w$ ($\#$ infected individuals). Above privacy cannot be guaranteed; below utility is not satisfied.}
    \label{fig:diff_priv_epsilon}
\end{figure}

If the health authority wants to release a heatmap to inform the public about hotspots or justify their policies, it must add additional noise to the map. Otherwise, the MNO could subtract the noise, which itself added in the first place, thus removing the protection provided by DP. The addition of noise by both parties does not violate privacy because DP enjoys composability~\cite{DBLP:journals/fttcs/DworkR14}. More concretely, if the heatmap produced by the MNO is $\epsilon_1$-differentially private and the health authority adds noise corresponding to $\epsilon_2$ to it, then the final heatmap is $(\epsilon_1+\epsilon_2)$-differentially private. The same methodology as above should be applied to choose $\epsilon_2$. It is crucial to find parameters such that the points $(w,\epsilon_1)$, $(w,\epsilon_2)$, and $(w,\epsilon_1+\epsilon_2)$ are in the plot's green area.

To illustrate the trade-off of \Cref{fig:diff_priv_epsilon} for a practical example, we performed experiments on the London subset of the publically available \textit{gowalla} dataset~\cite{gowalla}. This dataset consists of thousands of check-in's where each check-in consists of a user-id, GPS coodinates and a timestamp. To stay consistent with the methodology we discuss in the rest of the paper, we mapped all check-in locations to the locations of the nearest cell towers in London and treat multiple check-ins from the same user to the same cell tower as just one check-in. The final dataset consists of 4571 people and 9994 cell towers. \Cref{fig:heatmap} depicts a snipped of the original heatmap and the heatmaps resulting by applying DP, having $w=600$ randomly chosen infected people and varying $\epsilon$. The generated figures visually confirm our expectations based on the calculations above: One can observe that the heatmap without DP (\Cref{fig:heatmap_nodp}) is very similar to the heatmap with too little noise (\Cref{fig:heatmap_no_priv}), indicating that the noise is not enough to guarantee privacy. On the other hand, the heatmap with too much noise (\Cref{fig:heatmap_no_util}) clearly provides no utility due to the noise creating too many hotspots. In the correctly parameterized heatmap (\Cref{fig:heatmap_dp}), one can observe some difference to \Cref{fig:heatmap_nodp} due to noise, however the biggest hotspots remain the same. In other words, privacy and utility are preserved. Government officials now can use \Cref{fig:heatmap_dp} to set new policies (e.g., closing public locations in the hotspot areas) without the possibility to track the location of individuals.

\begin{figure*}[!htb]
    \begin{subfigure}[c]{0.5\textwidth}
        \centering
        \includegraphics[width=0.9\textwidth]{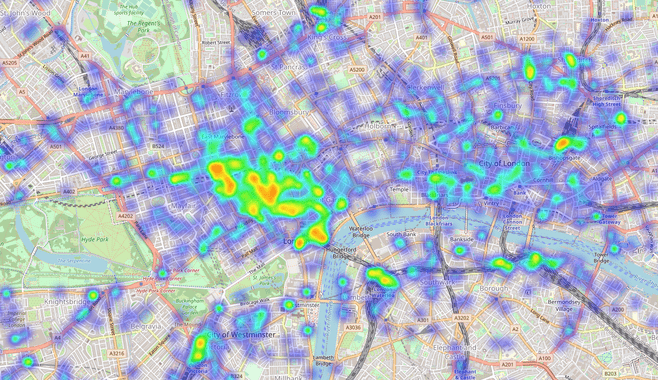}
        \subcaption{Original data without noise.}\label{fig:heatmap_nodp}

    \end{subfigure}%
    \begin{subfigure}[c]{0.5\textwidth}
        \centering
        \includegraphics[width=0.9\textwidth]{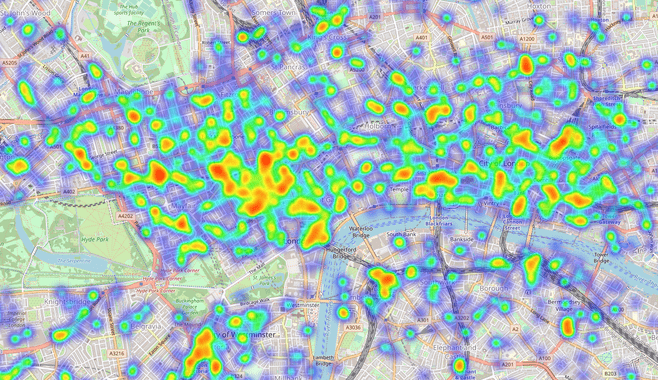}
        \subcaption{$\epsilon=0.05$: utility is not satisfied (red area in Fig. \ref{fig:diff_priv_epsilon}).}\label{fig:heatmap_no_util}

    \end{subfigure}\\[1em]
    \begin{subfigure}[c]{0.5\textwidth}
        \centering
        \includegraphics[width=0.9\textwidth]{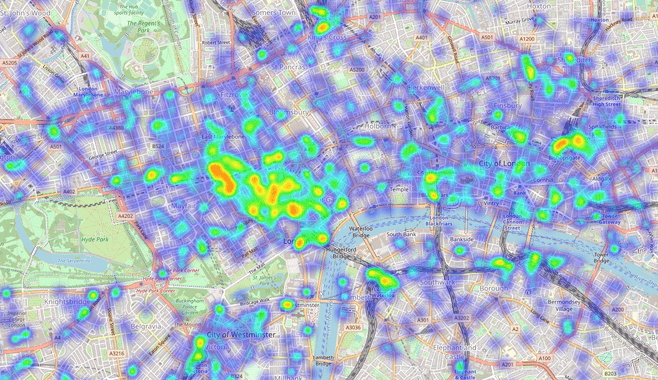}
        \subcaption{$\epsilon=0.6$: utility is satisfied and privacy is \\guaranteed (green area in Fig. \ref{fig:diff_priv_epsilon}).}\label{fig:heatmap_dp}

    \end{subfigure}%
    \begin{subfigure}[c]{0.5\textwidth}
        \centering
        \includegraphics[width=0.9\textwidth]{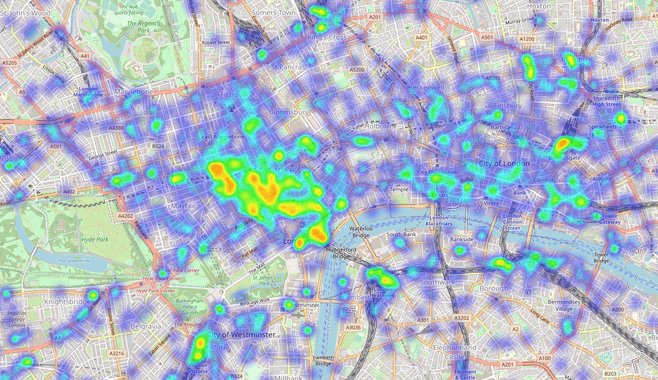}
        \subcaption{$\epsilon=3$: privacy is not guaranteed \\(above the privacy line in Fig. \ref{fig:diff_priv_epsilon}).}\label{fig:heatmap_no_priv}
    \end{subfigure}
    \caption{Influence of different $\epsilon$ values on an artificial heatmap created by mapping the \textit{gowalla}~\cite{gowalla} dataset onto Londoner cell towers. \Cref{fig:heatmap_nodp} shows the unmodified heatmap providing no privacy. While \Cref{fig:heatmap_no_priv} has too little noise for privacy (practically no difference compared to \Cref{fig:heatmap_nodp}), \Cref{fig:heatmap_no_util} has too much noise for utility. \Cref{fig:heatmap_dp} provides both privacy and utility. While the noise clearly influences the image, the hotspots remain the same.}
    \label{fig:heatmap}
\end{figure*}

\begin{remark}
    Several queries could contain the same individual. Since the overall movement pattern for the same individual changes slowly over time, we model this as an identical database. Therefore the total number of queries has to be limited to the total privacy budget. For example, if we follow the values of the analysis above and the health authority queries once a week for two months (= $8$ queries), the privacy budget suffices to provide utility as long as the number of infected individuals $w$ is above $750$ per week.
\end{remark}

\subsection{Summary and Limitations}
To summarize, the patients identifiers are encrypted during the whole protocol, hence, the semantic security of the HE scheme protects the privacy of the patients even if the MNO is cheating. The functional privacy of the HE scheme prevents, that the MNO's computation leaks anything about any location data to the health authority.
The binary mask guarantees that each individual is only present at most once in the query and prevents that a cheating health authority can amplify the contribution of individual's location data in the final heatmap. Differential privacy then prevents that location data from individual's can be singled out from the resulting heatmap. Consequently, the location data of individuals is protected even if the health authority is cheating. Hence, all sensitive information is always kept private from other parties during the whole protocol.

Even though privacy of input data is guaranteed, the protocol has some practical limitations. The protocol cannot guarantee, that either the health authority, or the MNO use truthful data in the first place. In other words, malicious health authorities can randomly mark individuals as infected and MNO's can use fake location data to create the heatmap. The protocol then would guarantee privacy of these wrong inputs, but the produced heatmap would be useless. This dependence on the truthfullness of the input data is, unfortunately, a generic problem to \textit{any} computation (plain and privacy preserving) and can not be prevented by cryptographic measures. We, therefore, propose that independent officials perform a yearly audit, e.g., at the end of the year, of the involved data to expose cheating parties.

Another limitation of our protocol is, that the utility of the heatmap scales with the prevalence of the disease. Concretely, the more people are infected, the smaller the impact of differential privacy on the final outcome. Conversely, the less people are infected the larger the impact of the noise and the utility drops. Thus, for very small prevalences it might not be possible to achieve high utility while maintainig privacy with our protocol.

\section{Implementation \& Performance}\label{sec:impl}
The data aggregation of our protocol requires only homomorphic ciphertext-ciphertext addition and homomorphic plaintext-ciphertext multiplication; however, the evaluation of the binary mask additionally requires homomorphic\\ciphertext-ciphertext multiplication. For our implementation we chose to use the BFV~\cite{DBLP:conf/crypto/Brakerski12,DBLP:journals/iacr/FanV12} SHE scheme, which fulfills these requirements. More specifically we use its implementation in the SEAL v3.6~\cite{sealcrypto} library, a fast, actively developed open-source library maintained by Microsoft Research.

The computationally most expensive phase in the protocol is the Data Aggregation phase, in which the MNO multiplies a huge matrix to a homomorphically encrypted input vector. Therefore, the main objective of our implementation is to perform this huge matrix multiplication as efficiently as possible.

\subsection{Packing} \label{sec:packing}
Modern HE schemes (including BFV) allow packing a vector of $n$ plaintexts into only one ciphertext. Performing an operation on this ciphertext then is implicitly applied to each slot of the encrypted vector, similar to single-instruction-multiple-data (SIMD) instructions on modern CPUs (e.g., AVX). However, the ciphertext size does not depend on the exact number ($\leq n$) of encoded plaintexts. The HE schemes support various SIMD operations, including slot-wise addition, subtraction and multiplication, and slot-rotation. However, one can not directly access a specific slot of the encoded vector. We can use the SIMD encoding to speed up the matrix multiplication of our protocol significantly.

In the BFV scheme (and its implementation SEAL), the number of available SIMD slots equals the degree of the cyclotomic reduction polynomial $(x^n+1)$; thus, it is always a power of two.
In the ciphertexts, the $n$ slots are arranged as matrix of dimensions $(2 \times \sfrac{n}{2})$. A ciphertext rotation affects either all rows or all columns of the matrix simultaneously. Therefore, we can think of the inner matrix as two rotatable vectors, which can be swapped.

\subsection{Homomorphic Matrix Multiplication}
Since SEAL does not provide algorithms for plain-matrix times encrypted vector multiplication, we implement the baby-step giant-step (BSGS) optimized matrix-vector multiplication~\cite{DBLP:conf/crypto/HaleviS14,DBLP:conf/eurocrypt/HaleviS15,DBLP:conf/crypto/HaleviS18} on our own and optimize it to fully leverage all slots (i.e., both rotatable vectors) of the homomorphic ciphertexts.

\subsubsection*{BSGS Matrix Multiplication.}
The SIMD encoding can be used to efficiently speed up matrix multiplication by using the diagonal method introduced by Halevi and Shoup~\cite{DBLP:conf/crypto/HaleviS14}, and its optimized version based on the BSGS algorithm~\cite{DBLP:conf/eurocrypt/HaleviS15,DBLP:conf/crypto/HaleviS18}:
\begin{align}
     & Z\cdot \vec{x}  = \sum_{i=0}^{m-1}\diag(Z,i)\circ \rot(\vec{x},i)          \label{eq:bsgs}                  \\
     & = \sum_{k=0}^{m_2-1} \rot\left(\sum_{j=0}^{m_1-1}\diag'(Z,km_1 + j)\circ \rot(\vec{x},j),km_1\right) \notag
\end{align}
where $m=m_1\cdot m_2$ and $\diag'(Z,i) = \rot\left(\diag(Z,i), -\floor{i/m_1}\cdot m_1\right)$.\footnote{In \Cref{eq:bsgs}, $\floor{i/m_1}$ is equal to $k$.} Note, that $\rot(\vec{x}, j)$ only has to be computed once for each $j<m_1$, therefore, \Cref{eq:bsgs} only requires $m_1 + m_2 - 2$ rotations of the vector $\vec{x}$ in total.

\subsubsection*{Extension to Bigger Dimensions.}
In our protocol, we want to homomorphically evaluate $\vec{x}^T \cdot Z = (Z^T\cdot\vec{x})^T$, where $\vec{x}\in\{0,1\}^N$ and $Z\in\Z_p^{N \times k}$, for big parameters $N$ and $k$. As described in \Cref{sec:packing}, the inner structure of the BFV ciphertext consists of two vectors of size $\sfrac{n}{2}$ each, and it does not allow a cyclic rotation over the whole input vector of size $n$. However, a rotation over the whole input vector is required by the BSGS algorithm. Therefore, we only can perform a BSGS multiplication with a $(\sfrac{n}{2} \times \sfrac{n}{2})$ matrix using this packing. Fortunately, we can use the remaining $\sfrac{n}{2}$ slots (i.e., the second vector in the inner structure of the BFV ciphertext) to simultaneously perform a second $(\sfrac{n}{2} \times \sfrac{n}{2})$ matrix multiplication. Therefore, after a homomorphic BSGS matrix multiplication, the result is a ciphertext $c$, where each of the two inner vectors encodes the result of a $(1\times \sfrac{n}{2}) \times (\sfrac{n}{2} \times \sfrac{n}{2})$ vector-matrix multiplication. The sum of those two vectors can easily be obtained by rotating the columns of the ciphertext $c$ and adding it to the first result:
\begin{equation}
    c_{sum} = c + \rot_\texttt{col}(c) \label{eq:finalsum}
\end{equation}
Thus, we can use one $(\sfrac{n}{2} \times \sfrac{n}{2})$ BSGS matrix multiplication and \Cref{eq:finalsum} to implement a homomorphic $(1\times n) \times (n \times \sfrac{n}{2}) = (1\times \sfrac{n}{2})$ vector-matrix multiplication.

Taking this into account, we split the huge $(N \times k)$ matrix into $n_v\cdot n_o$ submatrices of size $(n \times \sfrac{n}{2})$, with $n_v = \ceil{\frac{N}{n}}$ and $n_o = \ceil{\frac{2k}{n}}$, padding the submatrices with zeros if necessary. We split the input vector $\vec{x}$ into $n_v$ vectors of size $n$ (padding the last vector with zeros if necessary) and encrypt each of these vectors to get $n_v$ ciphertexts $c_i$. The final result of the $\vec{x}^T\cdot Z$ matrix multiplication can be computed with the following equation:
\begin{equation}
    \tilde{c}_i = \sum_{j=0}^{n_v - 1}\texttt{MatMul}(\texttt{SubMat}(Z, j, i)^T, c_j)\,\,\,\forall 0 \leq i < n_o \label{eq:splitmul}
\end{equation}
where, \texttt{SubMat}$(Z, j, i)$ returns the submatrix of $Z$ with size $(n \times \sfrac{n}{2})$, starting at row $n \cdot j$ and column $\frac{n}{2} \cdot i$, and \texttt{MatMul}$(Z, c)$ performs the homomorphic BSGS matrix multiplication $Z\cdot c$ followed by \Cref{eq:finalsum}.

\Cref{eq:splitmul} produces $n_o$ ciphertexts $\tilde{c}_i$, with the final results being located in the first $\sfrac{n}{2}$ slots of the ciphertexts. Overall, our algorithm to homomorphically calculate $\vec{x}^T \cdot Z$ requires $n_v\cdot n_o$ BSGS matrix multiplications and the total multiplicative depth is $1$ plaintext-ciphertext multiplication.

\subsection{Homomorphic Evaluation of the Mask}
To calculate the binary vector masking value (\Cref{eq:mask_bin}), we need to calculate the inner product of two homomorphically encrypted ciphertexts $c$ and $d$. After an initial multiplication $c\cdot d$, the inner product requires log$_2(\sfrac{n}{2})$ rotations and additions, followed by \Cref{eq:finalsum} to produce a ciphertext, where the result is encoded in each of the $n$ slots.
Our implementation uses rejection sampling and the \shake algorithm to cryptographically secure sample all the required random values in $\Z_p$. The total multiplicative depth to homomorphically evaluate the final mask (\Cref{eq:mask}) is $1$ ciphertext-ciphertext multiplication and $2$ plaintext-ciphertext multiplications.

\subsection{BFV parameters}
In BFV, one can choose three different parameters which greatly impact the runtime, security, and the available noise budget (i.e. how much further noise can be introduced until decryption will fail). These paramaters are the degree of the reduction polynomial $n=2^k$, the plaintext modulus $p$, which needs to be prime and $p \equiv$ $1$ (mod $2\cdot n$) to enable packing, and the ciphertext modulus $q$. We test our implementation for a computational security level of $\kappa=128\,$bit for different plaintext moduli $p$ using the smallest $n$ (and its default value for $q$) providing enough noise budget for correct evaluation of our protocol.

\subsection{Function Privacy and Noise Flooding}

Function privacy can be achieved by re-randomization and noise flooding, where the MNO adds an encryption of zero with a sufficiently large noise~\cite{DBLP:conf/stoc/Gentry09,DBLP:conf/eurocrypt/DucasS16} to the protocol's output. Following the smudging lemma~\cite{DBLP:conf/eurocrypt/AsharovJLTVW12}, one needs to add a ciphertext with noise being $\lambda_{\text{FP}} + \text{log}_2(n) + \text{log}_2(n_o)$ bits larger than the upper bound of our protocol's original output's noise to achieve a statistical distance of $2^{\lambda_{\text{FP}}}$ between different executions.

We implement noise flooding by creating an encryption of zero ($c_0$) with large noise (in practice, we set the noise as large as possible while ensuring decryption is still possible). Adding $c_0$ to the output of our protocol ($c$) results in a ciphertext which has $\lambda_{\text{FP}} = \textsc{noisebudget}(c) - \textsc{noisebudget}(c_0) - \text{log}_2(n) - \text{log}_2(n_o)$ statistical function privacy. In our concrete parameter sets, we ensure that $\lambda_{\text{FP}} > \nu$.

However, like most efficient instantiations of function privacy, noise flooding provides security against semi-honest adversaries only (see~\cite{DBLP:conf/eurocrypt/DucasS16} and contained references), and so our implementation also only provides semi-honest security.
Still, once available, our implementation can use efficient maliciously function-private FHE schemes instead and benefit from security against a malicious health authority.

\subsection{Benchmarks}

We benchmark our prototype implementation\footnote{The source code is available at \url{https://github.com/IAIK/CoronaHeatMap}.} on an c5.24xlarge AWS EC2 instance (96 vCPU @ 3.6\,GHz, 192\,GiB RAM) running Ubuntu Server 20.04 in the Region Frankfurt with a current price of $\$4.656$ per hour.

In our benchmarks, we focus on evaluating the runtime of the Data Aggregation phase of our protocol. Since in our use cases $N$ is much bigger than $k$, we implemented multithreading, such that the threads split the number of rows in the matrix (more specifically, the number of submatrices in the rows $n_v$) equally amongst all available threads. Therefore, each thread has to perform at most $\ceil{\frac{n_v}{\text{\#threads}}}\cdot n_o$ \texttt{MatMul} evaluations, which will be combined at the end by summing up the intermediate results.

The evaluation of the proving mask with its higher multiplicative depth requires BFV parameters providing a bigger noise budget, however, its actual evaluation does not impact the overall runtime of the protocol since we perform it in an extra thread in parallel to the data aggregation. Furthermore, adding DP, noise flooding, as well as the computations of the health authority (encryption and decryption), have negligible runtime.

The runtime of our protocol is $\mathcal{O}(n_v n_o)$, i.e., it scales linearly in the number of \texttt{MatMul} evaluations. This can be seen in \Cref{fig:runtime_dim} in which we summarize the runtime of the homomorphic matrix multiplication for different matrix dimensions using only one thread. For real-world matrix dimensions, some added runtime has to be expected due to thread synchronization and the accumulation of the intermediate thread results.

\begin{figure}[ht]
    \centering
    \begin{tikzpicture}
    \begin{axis}[%
        scatter/classes={%
                a={mark=pentagon*,red},%
                b={mark=square*,red},%
                c={mark=triangle*,green},%
                d={mark=square*,green},%
                e={mark=triangle*,blue},%
                f={mark=square*,blue},%
                g={mark=triangle*,red},%
                h={mark=pentagon*,blue},%
                i={mark=pentagon*,green}%
            },%
            ylabel={Runtime (sec)},
            xlabel={\#\texttt{MatMul}},
            legend style={legend pos=outer north east, draw=none,%
                cells={anchor=west, font=\scriptsize}},%
            ]
        ]

        \addlegendimage{empty legend}
        \addplot[scatter, scatter src=explicit symbolic]%
        table[meta=label] {
            x   y      label
            1      89.2 a
            5     443.1 b
            5     443.6 c
            10    886.1 d
            10    886.4 e
            20   1771.6 f
            20   1773.8 g
            25   2214.4 h
            50   4412.6 i
        };
        \legend{%
        \textbf{\hspace*{-0.6cm}Matrix Dimension $N\times k$},%
        $n\times \frac{n}{2} = 16384\times 8192$,%
        $n\times 5\frac{n}{2} = 16384\times 40960$,%
        $5n\times \frac{n}{2} = 81920\times 8192$,%
        $n\times 10\frac{n}{2} = 16384\times 81920$,%
        $10n\times \frac{n}{2} = 163840\times 8192$,%
        $20n\times \frac{n}{2} = 327680\times 8192$,%
        $n\times 20\frac{n}{2} = 16384\times 163840$,%
        $5n\times 5\frac{n}{2} = 81920\times 40960$,%
        $10n\times 5\frac{n}{2} = 163840\times 40960$,%
        }
    \end{axis}
\end{tikzpicture}
    \caption{Linear dependency of the runtime of the overall matrix multiplication to the number of \texttt{MatMul} evaluations. BFV parameters are: log$_2(p) = 42$, $n=16384$, $\kappa=128$.}
    \label{fig:runtime_dim}
\end{figure}

\subsubsection*{Real World Matrix Dimensions.}\label{sec:real_world_param}
In our benchmarks, we want to evaluate our protocol with parameters suitable for smaller nation states and set the matrix dimensions to $N$ being larger then the total population of small countries, and $k$ to be larger then the total number of cell towers in these countries. Concretely, we set $N=2^{23}$ and $k=2^{15}$, parameters enough to evaluate our protocol, for example, for Austria~\cite{austria1,austria2}, Singapore~\cite{enwiki:1002495194,singapore}, Kenya~\cite{wesolowski2012quantifying}, New York City, Paraguay or New Zealand. In \Cref{tab:benchmarks} we list the runtime for a homomorphic $(1\times 2^{23}) \times (2^{23}\times 2^{15})$ matrix multiplication, for different BFV parameters, using (at most) 96 threads. We also provide the total number of \texttt{MatMul} evaluations and the (maximum) number of evaluations per thread. We give performance numbers for a plaintext prime $p$ of size $42$\,bit, i.e., the smallest size to achieve $\nu=41$\,bit statistical privacy against malicious health authorities using our proving mask. To capture use cases, where a $42\,$bit plaintext modulus is not big enough, we also benchmark our protocol for a $60$\,bit prime $p$ (the maximum value supported by SEAL), providing $\nu=59\,$bit statistical security. Further, we also give the achieved statistical function privacy $\lambda_{\text{FP}}$ in bits for both benchmarks. As \Cref{tab:benchmarks} shows, the MNO's computation takes 70 minutes for a $42$\,bit plaintext prime and 1 hour 25 minutes for the bigger $60$\,bit prime.

\begin{table}[ht]
    \caption{Runtime for the MNO's computations for different parameters using $96$ threads. $N=2^{23}$, $k=2^{15}$, $\kappa=128$.}
    \label{tab:benchmarks}
    \centering
    \begin{tabular}{ccr|c|r|r}
        \toprule
        \multicolumn{3}{c|}{BFV} & \#\texttt{MatMul} & \centercellR{Time}    & \centercell{AWS}                                            \\
        log$_2(p)$               & \centercell{$n$}  & $\lambda_{\text{FP}}$ & $\,$total (thread) & \centercellR{min} & \centercell{price} \\
        \midrule
        $42$                     & $16384$           & $165$                 & $2048\,(24)$       & $69.33$           & \$$5.38$           \\ 
        $60$                     & $16384$           & $96$                  & $2048\,(24)$       & $83.23$           & \$$6.46$           \\ 
        \bottomrule
    \end{tabular}
\end{table}

\subsubsection*{Data Transmission.}
The data sizes which have to be transmitted between the MNO and the health authority are listed in \Cref{tab:communication}. Each row corresponds to a different parameter set from \Cref{tab:benchmarks}. The sizes were obtained by storing each of the described elements on the file system on the benchmarking platform. The table lists the size of the ciphertexts (ct), the public key (pk),  Galois keys (gk), and relinearization keys (rk). The public key is required for noise flooding to achieve function privacy, whereas Galois keys are required to perform homomorphic rotations. Each rotation index requires one Galois key, plus an additional key for rotating the columns. When using the BSGS algorithm, we need a key for the index $1$ to calculate $\rot(\vec{x}, j)$, and a key for the indices $k\cdot m_1,\, \forall 0 < k < m_2$. Also, for masking, we need the keys for the power-of-$2$ indices to calculate the inner product of two ciphertexts. The relinearization key is required to linearize the result of a ciphertext-ciphertext multiplication. We want to stress that the public key (pk), Galois keys (gk), and relinearization keys (rk) only need to be sent once before our protocol's first evaluation in a data-independent setup phase. Subsequent uses of the protocol can reuse these keys and only require transmitting the ciphertexts.

\begin{table}[ht]
    \centering
    \caption{Data transmission in MiB for parameters in \Cref{tab:benchmarks}.}
    \label{tab:communication}
    \begin{threeparttable}
        \begin{tabular}{ccccc|c|c}
            \toprule
            \multicolumn{5}{c|}{Health Authority} & MNO                              & Total                                                                                           \\
            ct                                    & \centercell{pk\tnotex{tnote:sp}} & \centercell{gk\tnotex{tnote:sp}} & rk\tnotex{tnote:sp} & \centercellR{Total} & ct    &          \\
            \midrule
            $445.9$                               & $1.0$                            & $557.5$                          & $7.8$               & $1012.2$            & $1.7$ & $1013.9$ \\
            $445.9$                               & $1.0$                            & $557.5$                          & $7.8$               & $1012.2$            & $1.7$ & $1013.9$ \\
            \bottomrule
        \end{tabular}
        \begin{tablenotes}
            \item\label{tnote:sp}One-time transmission (data-independent).
        \end{tablenotes}
    \end{threeparttable}
\end{table}

As \Cref{tab:communication} shows, health-authority-to-MNO communication is significantly more extensive than the response of the MNO. The main parts of the communication are the initial ciphertexts and the Galois keys.
One reason for the size difference between the ciphertexts in the query and the response is that the parameter $k$ is significantly smaller than $N$. Another reason is that our implementation performs a so-called modulus-switch after the computation, reducing the ciphertext modulus $q$ to only one of the moduli $q_i$ it is composed of. Further observe, that the plaintext modulus $p$ does not affect the communication cost.

\subsection{Price Estimation for Larger Countries}
Here we give an estimate of the costs of evaluating our protocol to create a COVID-19 heatmap for a larger country, more specifically, for Germany. About 83 million people live in Germany, and a total of 80000 cell sites are deployed~\cite{infozentrum}. With the BFV parameters of the first entry in \Cref{tab:benchmarks}, i.e., $n=16384,\nu=41, \kappa=128$, this corresponds to $n_v\cdot n_o=5066\cdot 10=50660$ \texttt{MatMul} evaluations.

To get $n_o = 10$ \texttt{MatMul} evaluations per thread, we would have to acquire $53$ CPU's capable of handling $96$ threads each. Assuming a runtime of $30\,$min per thread (calculated from \Cref{tab:benchmarks}), and a price of $\$4.656/h$ per CPU, we estimate the cost of evaluating the homomorphic matrix multiplication including the proving mask in a total time of $30\,$min to $\$124$ using AWS.\footnote{In practice, additional costs for handling the databases, network traffic, key management, human resources, among some other costs are to be expected.} This estimate shows that it is likely very feasible to create a heatmap once a week to gain valuable insight into the spread of the disease, even for larger countries. We, however, note that care has to be taken when outsourcing this computation to cloud providers to ensure user privacy in accordance to privacy regulations.

\section{Considerations and Conclusion}\label{sec:conlusion}

Our solution shows that privacy-preserving health data analytics is possible even on a national scale. We achieved this by combining three PETs. Each of them has their known limitations, but filtering out their strengths and applying them purposefully lead to a real-world cryptographic protocol. More broadly, we wanted to convey the following message: Even in times of crisis where it is tempting to lower data protection standards for purposes of big data analytics, there are technical methods to keep data protection standards high. And those technical methods are practical and available.

In the following we discuss considerations when instantiating our protocol with multiple health authorities or MNO's, as well as a summary of the key takeaways from a legal case study we conducted. More concretely, we focused on the EU General Data Protection Regulation (GDPR)~\cite{GDPR}, which is known to be on of the most strict privacy framework.

\subsubsection*{Multiple MNO's or Health Authorities.}
Even though it has already been shown that just using the largest MNO of a country for modelling disease dynamics is highly effective in practice~\cite{wesolowski2012quantifying}, one might consider to use data from multiple MNO's. Our protocol can easily be extended to this setting by performing the protocol with each MNO individually and summing up the resulting heatmaps. As long as DP parameters (\Cref{sec:epsilon}) are chosen, such that parameters $(w_i,\epsilon_i)$ for the $i$-th MNO, as well as $(\sum_i w_i, \sum_i \epsilon_i)$, fulfill the privacy-utility tradeoff, no additional information is leaked.

Multiple health authorities (e.g., for different provinces in a country) can be included using techniques from~\cite{DBLP:journals/popets/MouchetTBH21}. These multiple health authorities can agree on commen public keys, while keeping the decryption key hidden from all parties. After each health authority has agreed on database indices with the MNO (\Cref{rem:indices}), each authority can encrypt their queries using the common public key and the MNO can simply sum them up and proceed with the protocol as usual. After the protocol, the authorities proceed with the keyswitch protocol to output the final heatmap to some specified recipient (e.g., government officials). This adaptation is equivalent to the inital protocol with the same security and privacy guarantees, as long as each patient is registered with only one health authority. Otherwise, the heatmap will be a random output, due to the binary mask.

\subsubsection*{Legal Considerations.}
The health authority has used HE for COVID-19 positive individuals' ids, while the MNO has used DP to protected personal data. The MNO does not enter into possession of the decryption key of the health authorities data sets. Therefore, the computations performed should be considered carried out on anonymized data~\cite{spindler2016personal}, which are data that cannot identify, directly or indirectly, the data subject. In fact, data encrypted both by the health authority and the MNO is not accessible by an entity other than the one carrying out the encryption protocol. Hence, the data should be considered anonymized data -- whose processing falls out of the scope of application of the GDPR (Article 29 Working Party) -- for all other entities. A similar argument holds for the aggregated location data, which are protected from singling out attacks by DP~\cite{DBLP:journals/pnas/CohenN20,altman2020hybrid}.
Nevertheless, the processing of data by health authorities and MNO remains bound to GDPR provisions. In particular, the process to encrypt and make such data inaccessible is a processing activity under the GDPR. Thus, it should comply with legal requirements enshrined in the GDPR. In our use case, a lawful basis for processing personal data can be found in, e.g., Art. 9 (2) (i) GDPR, which deals with data processing in a public health context. It is one reasons why it is likely that there is a legal basis for our protocol.

Therefore, both MNO and health authority's processing activity protected through state-of-the-art PETs should be considered in compliance with GDPR provisions. From a legal perspective, the added value of the provided solution is represented on the one hand by the possibility to transform personal data into anonymized data. On the other hand, the processing activity of anonymizing data and limiting access to personal data ensure data subjects respect their fundamental rights as encoded in the EU privacy and data protection framework.

\section*{Acknowledgments}
We thank our shepherd Robert Cunningham for his constructive feedback and helpful insights for improving the paper. This work was supported by EU's Horizon 2020 project Safe-DEED under grant agreement \numero 825225, EU's Horizon 2020 project KRAKEN under grant agreement \numero 871473, and by the "DDAI" COMET Module within the COMET – Competence Centers for Excellent Technologies Programme, funded by the Austrian Federal Ministry for Transport, Innovation and Technology (bmvit), the Austrian Federal Ministry for Digital and Economic Affairs (bmdw), the Austrian Research Promotion Agency (FFG), the province of Styria (SFG) and partners from industry and academia. The COMET Programme is managed by FFG.

\bibliographystyle{splncs04}
\bibliography{bib,dblp}

\newpage

\appendix

\section{Security Proofs}\label{appendix:security}

We now prove security using the real-ideal-paradigm~\cite{DBLP:journals/ftsec/EvansKR18}. In this paradigm a protocol execution is secure if it behaves the same as when the parties send their input to a trusted third party (the ideal functionality) which does the computation and provides them with the outputs. More formally, an environment should not be able to distinguish between an observation of the protocol with a possible adversary and a simulator interacting with the ideal functionality. More specifically, most of the time, computational indistinguishability is required between the ideal and the real world. In contrast, we require $(\kappa,\nu)$-indistinguishability to analyze the cheating probability more thoroughly.

\begin{definition}[$(\kappa,\nu)$-indistinguishability~\cite{DBLP:journals/joc/LindellP12}]
    Let $X=\{X(a,\kappa,\nu)\}_{\kappa,\nu,\in\N,a\in\{0,1\}^*}$ and \\$Y=\{Y(a,\kappa,\nu)\}_{\kappa,\nu,\in\N,a\in\{0,1\}^*}$ be probability ensembles, so that for any $\kappa,\nu\in\N$ the distribution $\{X(a,\kappa,\nu)\}$ (resp. $\{Y(a,\kappa,\nu)\}$) ranges over strings of length polynomial in $\kappa+\nu$. We say that the ensembles are $(\kappa,\nu)$-indistinguishable if for every polynomial-time adversary $\mathcal{A}$, it holds that for every $a\in\{0,1\}^*$:
        \begin{equation*}
            \abs{\prob{\mathcal{A}\left(X=1\right)}-\prob{\mathcal{A}\left(Y=1\right)}}<\frac{1}{p(\kappa)}+2^{-\mathcal{O}(\nu)},
        \end{equation*}
        for every $\nu\in\N$, every polynomial $p(\cdot)$, and all large enough $\kappa\in\N$.
\end{definition}

\subsection{Binary Mask}\label{appendix:sec-masks}

\begin{lemma}\label{lem:masking_binary}
    Let $p$ be a integer of bit-length $\nu\in\N$, and let $N\leq 2^{\sfrac{\nu}{2}}$. Further, let $\vec{x}$ and $\mu_\mathtt{bin}$ be defined as in \Cref{sec:masking_binary}, then it holds that
    \[\prob{\vec{x}\;\mathtt{ not}\;\mathtt{binary}\;\land\;\mu_\mathtt{bin}=0}=\leq\frac{1}{2^{\nu-1}}.\]

\end{lemma}
\begin{proof}
    \begin{align*}
        \mu_\mathtt{bin}= & \underbrace{\langle \vec{x}, (\vec{d} \circ \vec{y_1}^N)\rangle \cdot r_1}_{:=\alpha} + \underbrace{\langle \vec{x}, (\vec{d} \circ \vec{y_2}^N) \rangle \cdot r_2}_{:=\beta}=\alpha+\beta
    \end{align*}
    We are now interested in the events when the binary mask evaluates to zero even though $\vec{x}\notin\Z_2^N$. This undesired behaviour can only happen in two ways, either $\alpha=\beta=0$ or $\alpha=-\beta$. Next, we calculate the probability of these two cases.

    First, since $r_1,r_2\neq 0$ and assuming $\vec{x}\neq \vec{0}^k$ ($\vec{x}= \vec{0}^k$ is a valid input and should result in a zero mask), we have $\prob{\alpha=0}=\prob{\beta=0}=\sfrac{N}{p}$~\cite{DBLP:conf/sp/BunzBBPWM18}. Hence,
    \begin{equation}\label{eq:prob_allzero}
        \prob{\alpha=\beta=0}=\frac{N}{p}\cdot\frac{N}{p}=\frac{N^2}{p^2}.
    \end{equation}
    Consequently, the probability of $\alpha$ being non-zero is $1-\sfrac{N}{p}$. Further, the probability of $\beta$ being $-\alpha$ is $\sfrac{1}{p}$. Combing these probabilities gives us
    \begin{equation}\label{eq:prob_inverse}
        \prob{\alpha=-\beta}=\left(1-\frac{N}{p}\right)\frac{1}{p}=\frac{1}{p}-\frac{N}{p^2}.
    \end{equation}
    We get the final probability by putting together \Cref{eq:prob_allzero} and \Cref{eq:prob_inverse}
    \begin{align*}
        \prob{\alpha+\beta=0} & =\frac{N^2}{p^2}+\frac{1}{p}-\frac{N}{p^2}<\frac{1}{p}+\frac{N^2}{p^2}                                               \\
                              & \leq\frac{1}{2^\nu}+\frac{2^\nu}{2^{2\nu}}=\frac{1}{2^{\nu-1}},\vspace{3mm}\text{ because }N\leq 2^{\sfrac{\nu}{2}}.
    \end{align*}
\end{proof}

\subsection{Proof of \Cref{lem:semi}}\label{appendix:proof_semi}

\begin{figure}[h]
    \centering
    \resizebox{0.95\columnwidth}{!}{
        \fbox{
            \parbox{\columnwidth}{
                \center{\textbf{\large{$\pi_{Hmap}$}}}
                \begin{enumerate}
                    \item A party $P_1$ on input $(\ifc{input},sid,P_1,P_2,\vec{x})$ from the environment verifies that $\vec{x}\in \Z^{N}_p$, else ignores the input. Next, samples a key pair $(\pk,\sk)\gets\kgen(1^{\kappa})$, and computes $\vec{c}\gets\enc_{\pk}(\vec{x})$. It records $(sid,P_1,P_2,\sk)$, and sends $(sid,P_1,P_2,\pk,\vec{c})$ to $P_2$. $P_1$ ignores subsequent inputs of the form $(\ifc{input},sid,P_1,P_2,\cdot)$ from the environment.
                    \item On a later input of the form $(sid,P_1,P_2,\vec{h}^*)$ from $P_2$, $P_1$ computes $\vec{h}\gets\dec_{\sk}(\vec{h}^*)$, and outputs $(\ifc{result},sid,P_1,P_2,\vec{h})$ to the environment.
                    \item A party $P_2$ on input $(\ifc{input},sid,P_1,P_2,Z)$ from the environment and $(sid,P_1,P_2,\pk,\vec{c})$ from $P_1$ verifies that $Z\in \Z^{N\times k}_p$, else ignores the input. Next, computes the mask vector $\vec{\mu}$ and the noise $\vec{\delta}$ according to \Cref{fig:final_protocol}. Then computes $\vec{h}^*\gets\eval_{\pk}(\vec{c}^T\cdot Z+\vec{\delta}+\vec{\mu})$. $P_2$, sends $(sid,P_1,P_2,\vec{h}^*)$ to $P_1$ and ignores all subsequent inputs of the form $(\ifc{input},sid,P_1,P_2,\cdot)$ from the environment.
                \end{enumerate}
            }
        }
    }
    \caption{Formalized protocol $\pi_{Hmap}$}
    \label{fig:protocol}
\end{figure}

\begin{figure}[h]
    \centering
    \resizebox{0.95\columnwidth}{!}{
        \fbox{
            \parbox{\columnwidth}{
                \center{\textbf{\large{$\simu_{Hmap}$}}}
                \begin{description}
                    \item[$P_1$, $P_2$ not corrupted:] It starts by sampling a key pair $(\pk,\sk)\gets\kgen(1^{\kappa})$, and sets $\vec{x}\gets 0^{N}$. Then it computes $\vec{c}\gets\enc_{\pk}(\vec{x})$. It then instructs $P_1$ to send $(sid,P_1,P_2,\pk,\vec{c})$ to $P_2$. On later input of the form $(sid,P_1,P_2,\pk,\vec{c})$ from $P_1$ it samples $Z\gets \Z_p^{N\times k}$. Then it computes $\vec{h}^*\gets
                              \eval_{\pk}(\vec{c}^T\cdot Z + \vec{\delta}+\vec{\mu})$. It instructs $P_2$ to send $(sid,P_1,P_2,\vec{h}^*)$ to $P_1$.
                    \item[$P_1$ not corrupted, $P_2$ corrupted:] Similar as before but it does not have to simulate $Z$ because it learns the input $Z$ from $P_2$. Then it computes $\eval_{\pk}(\vec{c}^T\cdot Z  + \vec{\delta}+\vec{\mu})$.
                    \item[$P_1$ corrupted, $P_2$ not corrupted:] It learns the input $\vec{x}$ from $P_1$. Then it proceeds as in the first case until it has to simulate the message to $P_1$. In order to do this it runs a copy of $\pi_{Hmap}$ internally, where it corrupts $P_1$. Thereby, it learns $\vec{x}^T\cdot Z + \vec{\delta}+\vec{\mu}$ and sets $\vec{h}^*\gets\enc_{\pk}(\vec{x}^T\cdot Z +\vec{\delta}+\vec{\mu})$.
                    \item[$P_1$, $P_2$ corrupted:] It learns the inputs $\vec{x}$ from $P_1$ resp. $Z$ from $P_2$. It runs the protocol with the inputs, and outputs $(\ifc{input},sid,P_1,P_2,\vec{x})$ and $(\ifc{input},sid,P_1,P_2,Z)$ to the ideal functionality, which makes $\mathcal{F}_{Hmap}$ output $(\ifc{result},sid,P_1,P_2,\vec{x}^T\!\!\cdot Z+\vec{\delta}+\vec{\mu})$.
                \end{description}
            }
        }
    }
    \caption{Simulator $\simu_{Hmap}$.}
    \label{fig:simulator}
\end{figure}

\begin{proof}
    We use \Cref{lem:masking_binary} to prove that to any polynomial time environment the execution $\pi_{Hmap}$ with a possible adversary $\mathcal{A}$ is $(\kappa,\nu)$-indistinguishable from a simulator $\simu$ interacting with the ideal functionality $\mathcal{F}_{Hmap}$. More concretely, we claim that as long as the event that $\vec{x}$ is not binary and at the same time the mask $\vec{\mu}=\vec{0}^k$ does not occur, the executions of the ideal and real world are computational indistinguishable. Once we have proven this claim, we are done, since we have already shown that the probability of the above event is exponentially small in the statistical security parameter. Note that for the proof, we have rewritten the protocol in a more formal description $\pi_{Hmap}$, see \Cref{fig:protocol}.

    First consider a polynomial time environment which does not corrupt any of the parties. Any meaningful environment will interact with $\pi_{Hmap}$ or $\mathcal{F}_{Hmap}$ in the following way.
    \begin{enumerate}
        \item It picks a vector $\!\vec{x}\!\in\!\Z^{n}_p$ and inputs $(\ifc{input},sid,P_1,P_2,\vec{x})$.
        \item It sees $(sid,P_1,P_2,\pk,\vec{c})$.
        \item It picks a matrix $Z\in \Z^{N\times k}_p$ and inputs $(\ifc{input},sid,P_1,P_2,Z)$.
        \item It sees $(sid,P_1,P_2,\pk,\vec{h}^*)$.
        \item It sees $(\ifc{result}, sid,P_1,P_2,\vec{h})$.
    \end{enumerate}
    Let us now assume to the contrary there is such an environment $\env$ that can distinguish the two systems $\pi_{Hmap}\circ\mathcal{A}$ and $\mathcal{F}_{Hmap}\circ\simu$ with non-negligible advantage. Then we can turn $\env$ into a polynomial time system $\env'$ which wins in the IND-CPA game with non-negligible probability:

    \begin{enumerate}
        \item $\env'$ receives $\pk$.
        \item $\env'$ runs $\env$ to see which message $(sid,P_1,P_2,\vec{x})$ gets recorded.
        \item $\env'$ inputs $(\vec{x},\vec{0}^{N})$ to the IND-CPA game and gets back an encryption $\vec{c}$, where $\vec{c}$ is either an encryption of $\vec{x}$ (if $b=0$) or an encryption of $\vec{0}^{N}$ (if $b=1$).
        \item $\env'$ samples $Z\gets \Z^{N}_p$. It runs $\env$ and provides input $(\ifc{input},sid,\allowbreak{}P_1,P_2,\vec{x})$, $(\ifc{input},sid,P_1,P_2,Z)$, $(sid,P_1,P_2,\pk,\vec{c})$, $(sid,P_1,P_2,\allowbreak{}\enc_{\pk}(\vec{c}^T\cdot Z + \vec{\delta}+\vec{\mu}))$ and $(\ifc{result},sid,P_1,P_2,\vec{x}^T\cdot Z + \vec{\delta}+\vec{\mu})$.
        \item $\env'\!$ waits until $\!\env\!$ outputs its guess $\!b'\!$, then $\!\env'\!$ outputs $\!b'\!$.
    \end{enumerate}

    If $b=0$, then $\env$ observes the interaction it would see when interacting with the protocol $\pi_{Hmap}$, and if $b=1$, then $\env$ observes the interaction it would see when interacting with the ideal functionality and the simulator $\mathcal{F}_{Hmap}\circ\simu$. By assumption $\env$ can distinguish $\pi_{Hmap}\circ\mathcal{A}$ and $\mathcal{F}_{Hmap}\circ\simu$ with non-negligible advantage. Therefore, $\env'$ will guess $b$ with probability significantly better than $\sfrac{1}{2}$. This is a contradiction to the IND-CPA security of $\he$, as $\env'$ is polynomial time.
\end{proof}

\subsection{One-Sided Simulation}

To define one-sided simulation security, we have the notion of a protocol execution view. Let $VIEW_{\pi,\mathcal{A}}^{\mathcal{A}}(x,y)$ denoted the protocol execution view of the adversary $\mathcal{A}$, i.e., the corrupted parties' view (input, randomness, all received messages) after execution of $\pi$ with input $x$ resp. $y$ from $P_1$ resp. $P_2$.

\begin{definition}
    Let $EXEC_{\pi,\mathcal{A},\mathcal{E}}$ resp. $EXEC_{\mathcal{F},\simu,\mathcal{E}}$ denote the random variables describing the output of environment $\mathcal{E}$ when interacting with an adversary $\mathcal{A}$ and parties $P_1$, $P_2$ performing protocol $\pi$, resp. when interacting with a simulator $\simu$ and an ideal functionality $\mathcal{F}$, where only $P_1$ receives output. Protocol $\pi$ securely realizes functionality $\mathcal{F}$ with one-sided simulation if
    \begin{enumerate}
        \item for any adversary $\mathcal{A}$ that controls $P_2$ there exists a simulator $\simu$ such that, for any environment $\mathcal{E}$ the distribution of $EXEC_{\pi,\mathcal{A},\mathcal{E}}$ and $EXEC_{\mathcal{F},\simu,\mathcal{E}}$ are indistinguishable,
        \item and for any adversary $\mathcal{A}$ controlling $P_1$ the distribution $\\VIEW_{\pi,\mathcal{A}}^{\mathcal{A}}(x,y)$ and $VIEW_{\pi,\mathcal{A}}^{\mathcal{A}}(x,y')$, where $|y|=|y'|$ are indistinguishable.
    \end{enumerate}
\end{definition}

\section{PSI-SUM with Indices} \label{app:psisum}
In \Cref{fig:ideal-psi-sum} we give the ideal functionality for a \textit{PSI-SUM with Indices} primitive. Such a primitve computes the sum of the private values associated with the intersection elements of two databases and reveals the indices present in the intersection to one party. This can be seen as an relaxed version of the \textit{Private Intersection-Sum with Cardinality} primitive introduced in~\cite{DBLP:conf/crypto/MiaoP0SY20}.

\begin{figure}[h]
    \centering
    \resizebox{0.95\columnwidth}{!}{
        \fbox{
            \parbox{\columnwidth}{
                \center{\textbf{\large{$\mathcal{F}_{PSI-Sum}$}}}\\
                \vspace{3mm}
                \raggedright
                Parameters: $t,N\in\N$. Parties $P_1$ and $P_2$.
                \begin{enumerate}
                    \item Upon receiving an input $(\ifc{input},sid,P_1,P_2,\vec{x})$ from a party $P_1$, verify that $\vec{x}\in \Z^{N}_p$, else ignore input. Next, record $(sid,P_1,P_2,\vec{x})$. Once $\vec{x}$ is recorded, ignore any subsequent inputs of the form $(\ifc{input},sid,P_1,P_2,\cdot)$ from $P_1$.
                    \item Upon receiving an input $(\ifc{input},sid,P_1,P_2,Z)$ from party $P_2$, verify that $Z\in \Z_p^{N\times *}$, else ignore input. Proceed as follows: If there is a recorded value $(sid,P_1,P_2,\vec{x},w)$, compute $\vec{h}\gets \vec{x}^TZ$ provided that $\vec{x}\in\Z_2^N$, otherwise $\vec{h}\stackrel{\$}{\gets}\Z^k_p$, and send $(sid,P_1,P_2,k)$ where $k$ is the number of columns of $Z$ to the adversary. Then output $(\ifc{result},sid,P_1,P_2,\vec{h})$ to $P_1$, and ignore subsequent inputs of the form $(\ifc{input},sid,P_1,P_2,\cdot)$ from $P_2$.
                \end{enumerate}
            }
        }
    }
    \caption{Ideal functionality of PSI-SUM with Indices.}
    \label{fig:ideal-psi-sum}
\end{figure}

\section{Differential Privacy}\label{appendix:diffpriv}

\begin{definition}[$\epsilon$-Differential Privacy~\cite{DBLP:conf/icalp/Dwork06}] \label{lem:diffpriv}
    A randomized mechanism $\mathcal{A}$ gives $\epsilon$-differential privacy if for any neighboring datasets $D$ and $D'$, and any $S \in Range(\mathcal{A})$: $Pr[\mathcal{A}(D) = S] \leq e^\epsilon Pr[\mathcal{A}(D') =S]$.
\end{definition}

One can achieve $\epsilon$-DP by adding noise from a zero-centered Laplace distribution to the final result of the computation. The noise is calibrated with the privacy budget $\epsilon$ and the global sensitivity $\Delta q$ of the computation $q$: $\Delta q = \max\limits_{D, D'}||q(D) - q(D')||$ for all neighboring $D$ and $D'$. The global sensitivity, thus, represents the maximum possible value of each element in the dataset. The Laplace distribution for a scale factor $b$ is given as $Lap(x|b) = \frac{1}{2b}e^{-\frac{|x|}{b}}$, where usually $b=\frac{\Delta q}{\epsilon}$.

\subsection{An economic method to choose \texorpdfstring{$\epsilon$}{epsilon}}\label{appendix:dp_ecnomic}
We aim to provide a confidence interval for the proportion $\mu$ of individuals in the general population (or subpopulation) with a specific property. Assume a database $D_N$ and let $g:D_N\to\R$ be the mechanism computing the sample mean with sensitivity $\sfrac{1}{N}$. If (for privacy reasons), we add Laplace noise $\nu$ to the outcome of $g$, we introduce an error source. Modeling each individual as a random variable with Bernoulli distribution allows us to bound this error by the tail bound. Hence, we can define the utility by a confidence interval with accuracy $T\in [0,1]$, and confidence $1-\alpha$ for $\alpha\in[0,1]$
\begin{equation*}
    \prob{\abs{g(D_N)+\nu(\epsilon)-\mu}\geq T}\leq e^{\left(-\frac{TN\epsilon}{2}\right)}\leq\alpha.
\end{equation*}
The idea of DP is that an individual's expected harm (cost) of being in the database should be minor. Let $E$ be the expected cost for an individual for being in the database (for a formal definition see~\cite{DBLP:conf/csfw/HsuGHKNPR14}). Then the individual's cost of being in the computation $g$ is
\[(e^{\epsilon}-1)E.\]
Let $D_{w}^j$ be the $j$-th column vector of the matrix $Z$, i.e., the location data corresponding to cell tower $j$. Then, we define the mechanism as
\[g(D_{w}^j):=\frac{\#\text{ individuals in }j}{w},\]
resulting in sensitivity $\sfrac{1}{w}$. This setup satisfies the assumption that each individual can be modeled as a Bernoulli experiment. This can be done for every cell tower, and thus covering the heatmap's area. The estimations of the expected baseline cost $E=\$0.01$ already cover the whole heatmap's area (all cell towers).

\end{document}